\documentclass[11pt,twoside]{article}
\pdfoutput=1

\usepackage{amsfonts}
\usepackage{amsmath}
\usepackage{amsthm}
\usepackage{amscd}
\usepackage{eucal}
\usepackage{dsfont}
\usepackage{bm}
\usepackage{graphicx}

\usepackage{a4}
\usepackage{times}
\usepackage{booktabs}
\usepackage{cite}
\usepackage{microtype}
\usepackage[title]{appendix}

\def\beq{\begin{equation}}
\def\eeq{\end{equation}}
\def\bea{\begin{eqnarray}}
\def\eea{\end{eqnarray}}

\def\beann{\begin{eqnarray*}}
\def\eeann{\end{eqnarray*}}

\let\a=\alpha  \let\g=\gamma \let\de=\delta
\let\e=\varepsilon   \let\th=\theta

  \let\la=\lambda \let\m=\mu
\let\n=\nu  \let\p=\pi \let\r=\rho \let\s=\sigma
\let\om=\omega

  \let\D=\Delta

\let\qd=\quad  

\def\epp{\, .}
\def\epc{\, ,}

\theoremstyle{plain}
\newtheorem{theorem}{Theorem}
\newtheorem*{theorem*}{Theorem}
\newtheorem{lemma}{Lemma}
\newtheorem*{lemma*}{Lemma}

\newtheorem*{corollary*}{Corollary}

\newtheorem*{conjecture*}{Conjecture}

\theoremstyle{definition}

\newtheorem*{remark}{Remark}
\newtheorem*{question*}{Question}

\renewcommand{\labelenumi}{(\roman{enumi})}

\def\2{\frac{1}{2}} \def\4{\frac{1}{4}}

\def\6{\partial}

\def\+{\dagger}

\def\<{\langle} \def\>{\rangle}

\let\then=\Rightarrow

\let\nodoti\i
\def\i{{\rm i}}

\def\rd{{\rm d}}

\DeclareMathOperator{\re}{e}

\DeclareMathOperator{\sh}{sh}
\DeclareMathOperator{\ch}{ch}

\DeclareMathOperator{\cth}{cth}

\DeclareMathOperator{\End}{End}
\DeclareMathOperator{\id}{id}

\DeclareMathOperator{\res}{res}

\def\Re{{\rm Re\,}} \def\Im{{\rm Im\,}}

\allowdisplaybreaks

\begin{document}

\thispagestyle{empty}

\begin{center}

{\Large \bf Dressed energy of the XXZ chain in the complex plane}

\vspace{10mm}

{\large
Saskia Faulmann,$^\dagger$
Frank G\"{o}hmann$^\dagger$ and
Karol K. Kozlowski$^\ast$}\\[3.5ex]
$^\dagger$Fakult\"at f\"ur Mathematik und Naturwissenschaften,\\
Bergische Universit\"at Wuppertal,
42097 Wuppertal, Germany\\[1.0ex]
$^\ast$Univ Lyon, ENS de Lyon, Univ Claude Bernard,\\ CNRS,
Laboratoire de Physique, F-69342 Lyon, France

\vspace{45mm}


{\large {\bf Abstract}}

\end{center}

\begin{list}{}{\addtolength{\rightmargin}{9mm}
               \addtolength{\topsep}{-5mm}}
\item
We consider the dressed energy $\varepsilon$ of the
XXZ chain in the massless antiferromagnetic parameter
regime at $0 < \Delta < 1$ and at finite magnetic field.
This function is defined as a solution of a Fredholm
integral equation of the second kind. Conceived as a
real function over the real numbers it describes the
energy of particle-hole excitations over the ground
state at fixed magnetic field. The extension of the
dressed energy to the complex plane determines the
solutions to the Bethe Ansatz equations for the
eigenvalue problem of the quantum transfer matrix of
the model in the low-temperature limit. At low
temperatures the Bethe roots that parametrize the
dominant eigenvalue of the quantum transfer matrix
come close to the curve ${\rm Re}\, \varepsilon
(\lambda) = 0$. We describe this curve and give lower
bounds to the function ${\rm Re}\, \varepsilon$ in
regions of the complex plane, where it is positive.
\end{list}

\clearpage

\section{Introduction}
The XXZ chain \cite{Orbach58, Walker59,YaYa66a,YaYa66b,YaYa66c}
is an anisotropic deformation of the Heisenberg chain \cite{Bethe31}.
It is the prototypical example of a Yang-Baxter integrable model
which is solvable by means of the algebraic Bethe Ansatz \cite{STF79}.
The Hamiltonian of the model acts on the tensor product space
${\cal H}_L = \bigotimes_{j=1}^L V_j$, $V_j = {\mathbb C}^2$,
in which every factor is identified with a lattice site in a
1d crystal. Expressed in terms of the familiar Pauli matrices
$\s^\a \in \End {\mathbb C}^2$, $\a = x, y, z$, the Hamiltonian
takes the form
\begin{equation} \label{hxxz}
     H_L = J \sum_{j = 1}^L \Bigl\{ \s_{j-1}^x \s_j^x + \s_{j-1}^y \s_j^y
                 + \D \bigl( \s_{j-1}^z \s_j^z - 1 \bigr) \Bigr\}
		 - \frac{h}{2} \sum_{j=1}^L \s_j^z \epp
\end{equation}
The three real parameters involved in this definition are the
anisotropy $\D$, the exchange interaction $J > 0$, and the
strength $h > 0$ of an external magnetic field.

The functions that characterize the properties of Yang-Baxter
integrable quantum systems in the thermodynamic limit, $L
\rightarrow \infty$, at zero temperature are defined as
solutions of Fredholm integral equations of the second kind
with kernels of difference form. The kernel functions
$K$ are given by the derivatives of the bare two-particle
scattering phases~$\th$ as functions of a rapidity variable
$\la$. If $S(\la)$ is the two-particle scattering factor for
a given~$\la$, then $S(\la) = \re^{2 \p \i \th (\la)}$ and
$K(\la) = \th' (\la)$.

For the XXZ chain with anisotropy parameter $\D = \cos(\g)$ we
have
\begin{equation}
     S(\la) = \frac{\sh(\la - \i \g)}{\sh(\la + \i \g)} \epp
\end{equation}
Hence, the kernel function is
\begin{equation} \label{defk}
     K(\la|\g) = \frac{1}{2 \p \i}
        \bigl( \cth(\la - \i \g) - \cth(\la + \i \g)\bigr) \epp
\end{equation}
In the following we restrict ourselves to the so-called
repulsive critical regime $0 < \D < 1$ corresponding to
$\g \in (0, \p/2)$.

We consider the integral equation
\begin{equation} \label{fredint}
     f(\la|Q) = f_0 (\la) - \int_{- Q}^Q \rd \m \; K(\la - \m|\g) f(\m|Q) \epc
\end{equation}
where $f_0 \in C^0 \bigl([- Q, Q]\bigr)$ will be called the
driving term.  It is not difficult to establish the existence
and uniqueness of solutions of (\ref{fredint}) on $C^0
\bigl([- Q, Q]\bigr)$. It follows from the convergence of the
Neumann series of the corresponding integral operator. The
proof and some further implications will be recalled below.

Once $f(\cdot|Q) \in C^0 \bigl([-Q,Q]\bigr)$ is fixed, the
integral on the right hand side of (\ref{fredint}) defines
a holomorphic $\i \p$-periodic function on the domain
\begin{equation}
     \Upsilon_\g (Q) =
        \bigl\{z \in {\mathbb C} \big|
	       z \notin [-Q,Q] \pm \i \g \mod \i \p\bigr\} \epp 
\end{equation}
If $f_0$ is meromorphic and $\i \p$-periodic on $\Upsilon_\g (Q)$,
then the same is true for $f(\cdot|Q)$ due to~(\ref{fredint}).
Functions defined this way play an important role in the study
of correlation functions of the XXZ chain in the zero-temperature
limit (see e.g.\ \cite{KMT99b,KKMST11b,DGK13a}).

The purpose of this work is to gain a better understanding of one
such special function, the dressed energy, on $\Upsilon_\g (Q)$.
Consider (\ref{fredint}) with driving term
\begin{equation}
     \e_0 (\la) = h - 4\p J \sin (\g) K(\la|\g/2) \epp
\end{equation}
This function is even on ${\mathbb R}$ and monotonically
increasing on ${\mathbb R}_+$.
\begin{equation}
     \min_{\la \in {\mathbb R}} \e_0 (\la) = \e_0 (0)
        = h - 4\p J \sin (\g) K(0|\g/2) = h - 4 J (1 + \D) \epp
\end{equation}
The condition $\e_0 (0) = 0$ determines the `upper critical
field'
\begin{equation}
     h_c = 4J (1 + \D) \epp
\end{equation}
Since $\lim_{\la \rightarrow \infty} \e_0 (\la) = h$, the function
$\e_0$ has a unique positive zero $Q_0$ if and only if
\begin{equation}
     0 < h < h_c \epp
\end{equation}
The latter condition defines the `critical parameter regime'. The
solution $\e (\la|Q)$ of (\ref{fredint}) with driving term $\e_0 (\la)$
has the following properties.

\begin{theorem} \label{th:fpoints}
Existence and uniqueness of Fermi points \cite{DGK14b}. Let
$\g \in (0, \p/2)$ and
\begin{equation}
     \e_u (\la) = h - \frac{2 \p J \sin (\g)}{\g \ch(\p \la/\g)} \epp
\end{equation}
\begin{enumerate}
\item
The function $\e(\la|Q)$ is a smooth function of $(\la,Q)$ on
${\mathbb R} \times (0, \infty)$ that is even in $\la$.
\item
For $\la \in {\mathbb R}$ it has the lower and upper bounds
\begin{subequations}
\label{epslowup}
\begin{align}
     \e_0 (\la) < \e(\la|Q) & & \text{for $0 < Q \le Q_0$,} \\[1ex]
     \e(\la|Q) < \e_u (\la) & & \text{for all $Q \ge 0$.}
\end{align}
\end{subequations}
\item
For any $h \in (0, h_c)$ exists a unique solution $Q_F > 0$ of
the equation $\e (Q|Q) = 0$. $Q_F$ is called the Fermi
rapidity.
\item
The Fermi rapidity is bounded by
\begin{subequations}
\label{qfbounds}
\begin{equation}
     Q_F < Q_0
\end{equation}
and, if there is a $Q_u$ with $\e_u (Q_u) = 0$ ($\Leftrightarrow
h < 2 \p J \sin(\g)/\g$), by
\begin{equation}
     Q_u < Q_F \epp
\end{equation}
\end{subequations}
\item
The function $h: (0, h_c) \rightarrow {\mathbb R}_+$, $h \mapsto
Q_F$ is smooth and monotonically decreasing with
$\lim_{h \rightarrow 0} Q_F = \infty$ and $\lim_{h \rightarrow h_c}
Q_F = 0$.
\end{enumerate}
\end{theorem}

\begin{remark}
The proof of this theorem given in \cite{DGK14b} is only valid
for $h < 2 \p J \sin(\g)/\g$, which is the condition for $Q_u$ to
exist. But it can be readily extended to the whole interval
$(0,h_c)$ (see below).
\end{remark}

We define the dressed energy by
\begin{equation} \label{defdressede}
     \e (\la) = \e (\la|Q_F) \epp
\end{equation}
A dressed energy function was introduced in the context of the
Bose gas with delta function interaction in \cite{YaYa69}. The
dressed energy (\ref{defdressede}) of the XXZ chain in the critical
regime first appeared \cite{TaSu72} in the low temperature limit
of the TBA equations that fix the thermodynamic properties of the
XXZ chain.

The dressed energy is a meromorphic $\i \p$-periodic function
on $\Upsilon_\g (Q_F)$ by construction. Alternatively, we may interpret
it as a function on the cylinder with cuts
\begin{equation}
     S_\g (Q_F) = \Upsilon_\g (Q_F) \cap
        \bigl\{z \in {\mathbb C} \big| - \p/2 \le \Im z < \p/2\bigr\} \epp
\end{equation}
By the implicit function theorem the equation
\begin{equation} \label{reepszero}
     \Re \e (\la) = 0
\end{equation}
determines a smooth curve on $S_\g (Q_F)$. This curve and the
functions $\Re \e$ and $\Im \e$ are further characterized by
the following theorem.
\begin{theorem} \label{th:main}
Dressed energy in the complex plane.
Let $\g \in (0,\p/2)$ and $\e$ be as in (\ref{defdressede}).
\begin{enumerate}
\item
For all $\la \in S_\g (Q_F)$ with $\Re \la = x$ and $\Im \la = y$
the function $\la \mapsto \Re \e (\la)$ is even in $x$ and in $y$.
\item
Within the strip $0 \le y < \g/2$ the function $x \mapsto
\Re \e (x + \i y)$ is monotonically increasing on ${\mathbb R}_+$
and, for every $y$, has a single simple zero $x (y)$.
\item
This determines a smooth function $x(y)$ on $(0,\g/2)$ which
behaves at the boundaries as $x(0) = Q_F$ and
\begin{equation} \label{leadx}
     x(y) \sim \sqrt{\frac{2 J \sin(\g)}{c} \biggl(\frac{\g}{2} - y\biggr)}
\end{equation}
with
\begin{equation}
     c = \frac{1}{1 - \frac{\g}{\p}}
         \biggl\{\frac{h}{2} + \int_{Q_F}^\infty \rd \m \: 
	         K \biggl(\frac{\m}{1 - \frac{\g}{\p}}\Big|
		          \frac{\g/2}{1 - \frac{\g}{\p}}\biggr)
			  \e(\m) \biggr\} > 0
\end{equation}
for $y \rightarrow (\g/2)_-$.
\item
Within the strip $|\Im \la| < \g/2$ the dressed energy is subject to
the bounds
\begin{equation}
     \Re \e_0 (\la) < \Re \e (\la) < \Re \e_u (\la) \epp
\end{equation}
\item
$\Re \e (\la) > 0$ for all $\la \in S_\g (Q_F)$ with $|\Im \la| > \g/2$,
and we have the lower bounds
\begin{subequations}
\label{reepsgreater}
\begin{align}
     & \Re \e(\la) > h & & \text{if} \qd
                  \frac{\p}{2} - \2 \biggl(\frac{\p}{2} - \g\biggr)
		  < y < \frac{\p}{2} \epc \label{reeps4} \\[1ex]
     & \Re \e(\la) > \frac{h}{2} & & \text{if} \qd
		  \g < y < \frac{\p}{2} - \2 \biggl(\frac{\p}{2} - \g\biggr)
		  \epc \label{reeps3} \\[1ex]
     & \Re \e(\la) > \min \biggl\{\frac{h}{2}, \frac{h \g}{\p - \g}\biggr\} & & \text{if} \qd
		  \frac{\g}{2} < y < \g \epp \label{reeps2}
\end{align}
\end{subequations}
\item
For all $\la \in S_\g (Q_F)$ with $\Re \la = x$ and $\Im \la = y$
the function $\la \mapsto \Im \e (\la)$ is odd in $x$ and in $y$.
\item
$\Im \e$ is monotonically increasing along the curve $x(y)$,
\begin{equation}
     \frac{\rd \, \Im \e (x(y) + \i y)}{\rd y} > 0 \epc
\end{equation}
$\Im \e(x(0)) = 0$ and
\begin{equation} \label{leadim}
     \Im \e(x(y) + \i y) \sim \sqrt{\frac{2 J \sin(\g) c}{\g/2 - y}}
\end{equation}
for $y \rightarrow (\g/2)_-$.
\end{enumerate}
\end{theorem}
This theorem is our main result. It will be proven below.
Examples of the curve (\ref{reepszero}) for various sets of
parameters are shown in Fig.~\ref{fig:eps_curve_various_mag}. Our
\begin{figure}
\begin{center}
\includegraphics[width=.96\textwidth]{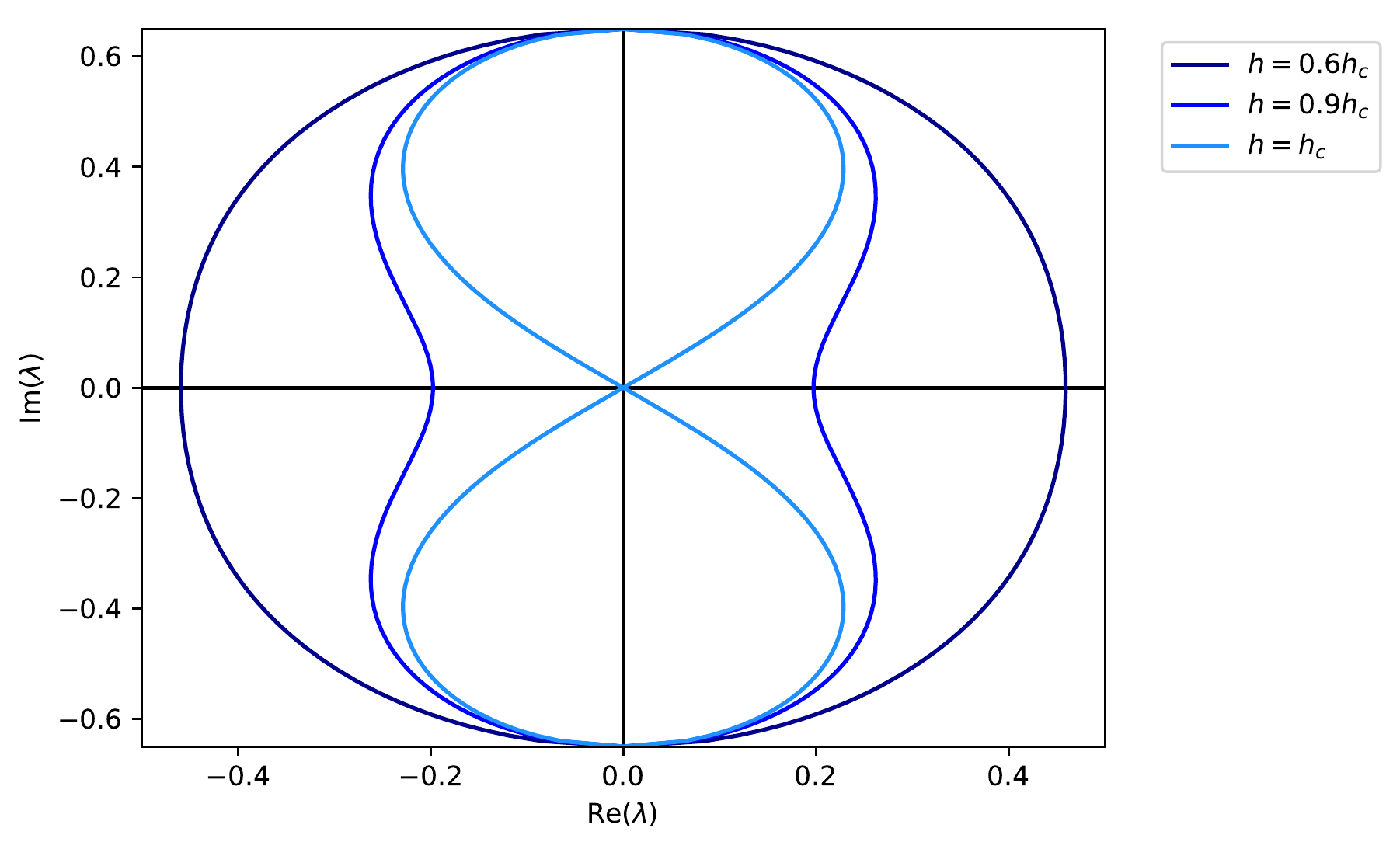}
\end{center}
\caption{\label{fig:eps_curve_various_mag}
The curves (\ref{reepszero}) for $J = 1$, $\g = 1.3$ and various
values of the magnetic field in units of $h_c = 5.07$. Loosely
speaking, Theorem~\ref{th:main} says that the figure describes
the generic situation: $\Re \e (\la) = 0$ is a simple closed
curve, situated entirely inside the strip $|\Im \la| < \frac{\g}{2}$,
symmetric with respect to the real and imaginary axis, such that
its positive part, $\Re \la > 0$, is the graph $\{x(y) + \i y|
y \in (- \g/2,\g/2)\}$ of a smooth function $x(y)$. At $h = h_c$
this graph develops a cusp which signals the transition to the
fully polarized massive regime.
}
\end{figure}
interest in the curve (\ref{reepszero}) and in the estimates
(\ref{reepsgreater}) comes from our work on thermal form factor
series for the correlation functions of the XXZ chain
(see e.g.\ \cite{DGK13a,DGKS16b,GKKKS17,BGKS21a}). The
derivation of the series requires knowledge of the full spectrum
of the quantum transfer matrix \cite{Suzuki85,SAW90,Kluemper92}
of the model. So far we have found a characterization of the
full spectrum only in the massive antiferromagnetic regime
($\D > 1$ and $0 < h < h_\ell$, where $h_\ell$ is a lower
critical field) in the low-temperature limit \cite{DGKS15b}.
This case is characterized by the absence of so-called string
excitations. Theorem~\ref{th:main} will be needed in order to
establish a similar behaviour in the massless regime. This
is what we would like to achieve in a subsequent paper. It
will be dealing with the low-temperature analysis of the
auxiliary functions and the spectrum of the quantum transfer
matrix of the XXZ chain for $- 1 < \D < 1$. The `critical part
of the spectrum', pertaining to excitations about the two Fermi
points $\pm Q_F$, was analyzed in \cite{DGK13a,DGK14a}. In our
forthcoming work we want to exclude the existence of strings
in the low-temperature limit. This will show that not only the
Bethe roots of the dominant state, but the Bethe roots belonging
to any Bethe eigenstate of the quantum transfer matrix come
close to the curve $\Re \e (\la) = 0$, when the temperature
goes to zero. The latter will then be a crucial input for the
further investigation of the thermal form factor series of the
two-point functions of the XXZ chain in the critical regime.

Our two theorems above are stated for a restricted parameter
regime, $\g \in (0,\p/2)$. This has several reasons. First of
all we wanted to avoid further case distinctions in order
to keep this work reasonably short and reader-friendly. In
fact a version of Theorem~\ref{th:fpoints} valid for $\g \in
(\p/2,\p)$ can be found in \cite{DGK14b}. As for the extension
of Theorem~\ref{th:main}, the techniques developed in
\cite{DGK14b} and below can be used. There are, however, certain
technical difficulties which come from the fact that for
$\g > 2\p /3$ the pole of the driving term $\e_0$ is beyond the
cuts caused by the poles of the kernel function inside the fundamental
cylinder, which are, in this case, located at $[-Q_F, Q_F] \pm
\i (\p - \g)$. These problems can be dealt with by a deformation
of the integration contour in the integral equation (\ref{fredint}),
but this is more naturally done in conjunction with the
low-$T$ analysis of the non-linear integral equations for the
auxiliary functions. We would also like to point out that
for some of the proofs of the properties of $\e$ for $\g \in
(0,\p/2)$ we will need to know the properties of the kernel
function for $\g \in (0,\p)$, which is why Lemma~\ref{lem:propk}
below is formulated for the extended parameter region.

\section{Preliminaries}
\subsection{Properties of the kernel function}
At several instances we will use Fourier transformation techniques.
Our convention for the Fourier transform of a function $f:
{\mathbb C} \rightarrow {\mathbb C}$ is
\begin{equation}
     {\cal F} [f] (k) = \int_{- \infty}^\infty \rd \la \: \re^{\i k \la} f(\la) \epp
\end{equation}

\begin{lemma} \label{lem:propk}
Properties of the kernel function.
\begin{enumerate}
\item
$K(\cdot|\g)$ defines a smooth even function on ${\mathbb R}$ which
is monotonously decreasing on ${\mathbb R}_+$ if $0 < \g < \p/2$
and monotonously increasing on ${\mathbb R}_+$ if $\p/2 < \g < \p$.
\item
$K(\la|\g) > 0$ for all $\la \in {\mathbb R}$ if $0 < \g < \p/2$, and
$K(\la|\g) < 0$ for all $\la \in {\mathbb R}$ if $\p/2 < \g < \p$.
\item
$K(\cdot|\g)$ is meromorphic on $S_\g (Q)$ with two simple poles
which are located at $\pm \i \g$ if $0 < \g < \p/2$ or at
$\pm \i (\p - \g)$ if $\p/2 < \g < \p$.
\item
For $x, y \in {\mathbb R}$
\begin{equation} \label{rekkk}
     \Re K(x + \i y|\g) = \2 \bigl(K(x|\g - y) + K(x|\g + y)\bigr) \epc
\end{equation}
implying that $\Re K(x + \i y|\g)$ is an even function of $x$ for
fixed $y$ and an even function of $y$ for fixed $x$.
\item
\begin{equation} \label{kfourier}
     {\cal F} [K(\cdot|\g)] (k) = \frac{\sh((\p/2 - \g)k)}{\sh(\p k/2)} \epp
\end{equation}
\end{enumerate}
\end{lemma}
\begin{proof}
The kernel function $K(\la|\g)$ can be rewritten as
\begin{equation} \label{realk}
     K(\la|\g) = \frac{\sin(2\g)}{2 \p \bigl(\sh^2(\la) + \sin^2(\g)\bigr)}
\end{equation}
from which we can read of (i) and (ii). (iii) and (iv) are
direct consequences of the definition (\ref{defk}). The
calculation of the Fourier transform (v) is a standard
exercise using the $\i \p$-periodicity of $K(\cdot|\g)$ and
the residue theorem.
\end{proof}

\subsection{\boldmath The solvable case $Q = \infty$}
For $Q = \infty$ the integral equation (\ref{fredint})
can be solved by means of Fourier transformation and the
convolution theorem. This gives us explicit solutions for
various driving terms $f_0$. As we shall see, some of these
play an important role as bounds for the general case of
finite $Q$. The most important such function is the resolvent
kernel $R(\cdot|\g)$. It is the solution of (\ref{fredint})
for $Q = \infty$ and with driving term $f_0 (\la) =
K(\la - \m|\g)$.

\begin{lemma}
Properties of the resolvent kernel for $Q = \infty$
\cite{YaYa66c}.
\begin{enumerate}
\item
The resolvent kernel $R(\cdot|\g)$ has the Fourier integral
representation
\begin{equation} \label{rfourier}
     R(\la|\g) =
        \int_{- \infty}^\infty \frac{\rd k}{4 \p}
	\frac{\re^{- \i k \la} \sh \bigl((\p/2 - \g) k\bigr)}
	     {\ch (\g k/2) \sh \bigl((\p - \g) k/2\bigr)} \epc
\end{equation}
valid for $|\Im \la| < \g$.
\item
For $0 < \g < \p/2$, $R(\cdot|\g)$ has the convolution type
representation
\begin{equation} \label{rconv}
     R(\la|\g) =
        \frac{\p}{2 \g(\p - \g)} \int_{- \infty}^\infty \rd \m \:
	\frac{K \Bigl(\frac{\m}{1 - \frac{\g}{\p}}\Big|\frac{\g/2}{1 - \frac{\g}{\p}}\Bigr)}
	     {\ch \bigl((\la - \m) \frac{\p}{\g}\bigr)} \epc
\end{equation}
valid for $|\Im \la| < \g/2$.
\item
For $0 < \g < \p/2$, $R(\cdot|\g)$ is even and positive on
${\mathbb R}$ and monotonically decreasing on ${\mathbb R}_+$,
where it satisfies $\lim_{\la \rightarrow \infty} R(\la|\g) = 0$.
\end{enumerate}
\end{lemma}
\begin{proof}
(i) Fourier transforming the integral equation
\begin{equation} \label{intr}
     R(\la|\g) = K(\la|\g) - \int_{- \infty}^\infty \rd \m \: K(\la - \m|\g) R(\m|\g)
\end{equation}
and solving for ${\cal F} [R(\cdot|\g)]$ we obtain
\begin{equation}
     {\cal F} [R(\cdot|\g)] (k) = 
	\frac{\sh \bigl((\p/2 - \g) k\bigr)}
	     {2 \ch (\g k/2) \sh \bigl((\p - \g) k/2\bigr)} \epp
\end{equation}
For $k \rightarrow \pm \infty$ we see that ${\cal F} [R(\cdot|\g)] (k)
\sim \re^{- \g|k|}$, implying that the back transformation
(\ref{rfourier}) converges for all $\la$ with $|\Im \la| < \g$.

(ii) The convolution type representation is obtained from
(\ref{rfourier}) by rescaling $k \rightarrow k/(1 - \g/\p)$ and
setting
\begin{equation} \label{defgprime}
     \g' = \frac{\g/2}{1 - \g/\p} \epp
\end{equation}
Then
\begin{equation}
     R(\la|\g) = \frac{1}{1 - \g/\p} \int_{- \infty}^\infty \frac{\rd k}{4 \p}
                 \frac{\re^{- \i \frac{k \la}{1 - \g/\p}}}{\ch(\g' k)}
		 {\cal F}[K(\cdot|\g')] (k) \epc
\end{equation}
which implies (\ref{rconv}) by employing the convolution theorem
on the right hand side. Note that $\g \mapsto \g'$ is a
monotonically increasing function that maps $(0,\p/2) \rightarrow
(0,\p/2)$. Because of the poles of $K(\cdot/(1 - \g/\p)|\g')$ and
$1/\ch(\cdot\: \p/\g)$ at $\pm \i \g/2$ the validity of the
representation (\ref{rconv}) is restricted to $|\Im \la| < \g/2$.

(iii) From the representation (\ref{rconv}) it is clear that
$R(\la|\g) > 0$ and that $R(\la|\g)$ is even in~$\la$. Both,
$K(\la/(1 - \g/\p)|\g')$ and $1/\ch(\la \p/\g)$, are even, positive,
integrable over ${\mathbb R}$, and go to zero monotonically for
$\la \rightarrow \pm \infty$. For any two kernels $K_1$, $K_2$
with these properties and all $\la > 0$ we have the estimate
\begin{multline}
     \int_{- \infty}^\infty \rd \m \: K_1 (\la - \m) K_2 (\m) \\
        < K_1 (\la/2) \int_{- \infty}^{\la/2} \rd \m \: K_2 (\m)
	  + K_2 (\la/2) \int_{\la/2}^\infty \rd \m \: K_1 (\la - \m) \epp
\end{multline}
Hence, (\ref{rconv}) implies that $\lim_{\la \rightarrow \pm \infty}
R(\la|\g) = 0$.

Furthermore,
\begin{multline} \label{rprimereal}
     R'(\la|\g) =
        \frac{1}{2\g (1 - \g/\p)^2} \int_0^\infty \rd \m \:
	   K' \Bigl(\frac{\m}{1 - \g/\p}\Big|\g'\Bigr) \\ \times
	   \biggl\{\frac{1}{\ch\bigl((\la - \m)\frac{\p}{\g}\bigr)} -
	           \frac{1}{\ch\bigl((\la + \m)\frac{\p}{\g}\bigr)}\biggr\} \epp
\end{multline}
Now $1/\ch$ is even and monotonically decreasing on ${\mathbb R}_+$,
and $|\la - \m| < \la + \m$ for all $\la, \m \in {\mathbb R}_+$,
implying that the term in the curly brackets under the integral
is positive. Since $K'(\m/(1 - \g/\p)|\g') < 0$ for $\m > 0$, it
follows that $R'(\la|\g) < 0$ for all $\la > 0$.
\end{proof}
Two more `dressed functions' for $Q = \infty$ that will be needed
below are $\e_\infty$, the solution of (\ref{fredint}) with $Q =
\infty$ and $f_0 = \e_0$, and $\r_\infty$, the solution of
(\ref{fredint}) with $Q = \infty$ and $f_0 = K(\cdot|\g/2)$.
Using the convolution theorem we see that
\begin{subequations}
\begin{align}
     \e_\infty (\la) & =
        \frac{h}{2\bigr(1 - \frac{\g}{\p}\bigl)}
	- \frac{2 \p J \sin(\g)}{\g \ch\bigl(\frac{\p\la}{\g}\bigr)} \epc \\
     \r_\infty (\la) & =
	\frac{1}{2 \g \ch\bigl(\frac{\p\la}{\g}\bigr)} \epp
\end{align}
\end{subequations}

\subsection{\boldmath The general case of finite $Q$}
The existence of a unique solution of (\ref{fredint}) can be
established by standard arguments. Consider the linear integral
operator $\hat K: C^0 \bigl([-Q,Q]\bigr) \rightarrow
C^0 \bigl([-Q,Q]\bigr)$ defined by
\begin{equation}
     \hat K f (\la) = \int_{-Q}^Q \rd \m \: K(\la - \m|\g) f(\m)
\end{equation}
for $C^0 \bigl([-Q,Q]\bigr)$ equipped with the sup-norm $\|\cdot\|_\infty$.
Then
\begin{multline} \label{kcontract}
     \|\hat K\| = \sup_{f \in C^0 ([-Q,Q])} \frac{\|\hat K f\|_\infty}{\|f\|_\infty}
        \le \max_{\la \in [-Q,Q]} \int_{-Q}^Q \rd \m \: K(\la - \m|\g) \\
	< \int_{- \infty}^\infty \rd \m \: K(\m|\g) = {\cal F}[K(\cdot|\g)] (0)
	= 1 - \frac{2 \g}{\p} < 1 \epc
\end{multline}
which proves the convergence of the series $\sum_{n=0}^\infty (- \hat K)^n =
\bigl(\id + \hat K\bigr)^{-1}$.

The resolvent kernel $R_Q (\la, \m)$ is the solution of
(\ref{fredint}) with $f_0 (\la) = K (\la - \m|\g)$.
In our notation $R_Q (\la,\m)$ we suppress the parametric
dependence of $R_Q$ on $\g$, since it will be fixed throughout
this work. $R_Q$ has the following properties.
\begin{lemma} \label{lem:reskerqfinite}
Resolvent kernel at finite $Q$ \cite{DGK14b}.
\begin{enumerate}
\item
$R_Q (\cdot, \m)$ is meromorphic on $S_\g (Q)$ with simple
poles at $\m \pm \i \g$ and depends smoothly on $Q \in
{\mathbb R}_+$.
\item
The integral operator associated with $R_Q$ commutes with
$\hat K$,
\begin{equation} \label{resolvcom}
     \int_{-Q}^Q \rd \n \: K(\la - \n|\g) R_Q (\n, \m)
        = \int_{-Q}^Q \rd \n \: R_Q (\la, \n) K(\n - \m|\g) \epp
\end{equation}
\item
$R_Q (\la, \m) = R_Q (\m, \la)$ and $R_Q (\la, \m) = R_Q (- \la, - \m)$.
\end{enumerate}
\end{lemma}
\begin{proof}
(i) It follows from (\ref{kcontract}) that the spectral radius of
$\hat{K}$ is strictly less than one. Hence, its Fredholm determinant
\begin{equation}
     \det \bigl[\id + \hat{K} \bigr]
        = \sum_{n \geq 0} \frac{1}{n!}
	  \int_{-Q}^Q \rd^n \: \nu
	  \det_n \bigl[K(\nu_a - \nu_b |\g)\bigr]
\end{equation}
does not vanish, uniformly in $Q > 0$ and is bounded. Clearly, it
is also a smooth function of $Q$. The resolvent kernel $R_{Q}(\la,\m)$
is given by the below, absolutely convergent, series of multiple
integrals, see e.g.\ \cite{GGK00},
\begin{equation}
     R_Q (\la,\m) = 
        \frac{1}{\det \bigl[\id + \hat{K} \bigr]}
        \sum_{n \geq 0} \frac{1}{n!}
	\int_{-Q}^Q \rd^n \nu \: \det_{n+1}
	   \begin{bmatrix}
	      K(\la - \m|\g) & K(\la - \nu_b|\g)  \\
	      K(\nu_a - \m|\g)& K(\nu_a - \nu_b|\g) 
           \end{bmatrix} \epp
\end{equation}
This readily entails that $\la \mapsto R_Q (\la,\m)$ is
meromorphic on $S_\g (Q)$ with simple poles at $\mu \pm \i \g$.
Since each summand of the above absolutely convergent series
is a smooth function of $Q$ belonging to compact subsets of
${\mathbb R}_+$, the same follows for the resolvent kernel.

(ii) Consider a kernel $\overline{R}_Q (\la, \m)$ defined as the
unique solution of the integral equation
\begin{equation} \label{leftresolvent}
     \overline{R}_Q (\la, \m) =
        K(\la - \m|\g) - \int_{-Q}^Q \rd \n \: \overline{R}_Q (\la, \n)
	                    K(\n - \m|\g) \epp
\end{equation}
Using this equation and the integral equation for $R_Q (\cdot,\m)$,
we see that
\begin{multline}
     \int_{-Q}^Q \rd \n \: \overline{R}_Q (\la, \n) R_Q (\n, \m) +
        \int_{-Q}^Q \rd \n_1 \: \int_{-Q}^Q \rd \n_2 \:
	   \overline{R}_Q (\la, \n_1) K(\n_1 - \n_2|\g) R_Q (\n_2, \m) \\
        = \int_{-Q}^Q \rd \n \: K(\la - \n|\g) R_Q (\n, \m)
        = \int_{-Q}^Q \rd \n \: \overline{R}_Q (\la, \n) K(\n - \m|\g) \epp
\end{multline}
Substituting the last equation into (\ref{leftresolvent}) and
comparing with the defining integral equation for $R_Q (\cdot,\m)$
we conclude that $\overline{R}_Q (\la, \m) = R_Q (\la, \m)$, which
proves the claim.

The first statement of (iii) follows by interchanging $\la$ and
$\m$ in the defining integral equation for $R_Q (\la, \m)$,
then using (\ref{resolvcom}) and the uniqueness of the solution
of the integral equation. Using the uniqueness also the second
statement follows by negating $\la$ and $\m$ in the defining
integral equation and exploiting that $K (\cdot|\g)$ is even.
\end{proof}

Every solution of (\ref{fredint}) with a driving term $f_0$ that
is uniformly bounded on ${\mathbb R}$ satisfies a second linear
integral equation \cite{YaYa66c} with respect to the complementary contour
${\mathbb R} \setminus [-Q,Q]$. By definition $f_\infty$ is
the solution of the integral equation
\begin{equation}
     f_\infty (\la) = f_0 (\la)
        - \int_{-\infty}^\infty \rd \m \:
	  K(\la - \m|\g) f_\infty (\m) \epp
\end{equation}
If $f_0$ is uniformly bounded on ${\mathbb R}$, then the
same holds for $f$ as follows from (\ref{fredint}), and
\begin{equation}
     f(\la) = f_0 (\la) +
        \int_{{\mathbb R} \setminus [-Q,Q]} \rd \m \: K(\la - \m|\g) f(\m)
	- \int_{- \infty}^\infty \rd \m \: K(\la - \m|\g) f(\m) \epp
\end{equation}
Conceiving this equation as an integral equation on the real axis
with driving term $f_0 (\la) + \int_{{\mathbb R} \setminus [-Q,Q]} \rd \m \:
K(\la - \m|\g) f(\m)$ and using its linearity we obtain
\begin{equation} \label{fredintcompl}
     f(\la) = f_\infty (\la) +
        \int_{{\mathbb R} \setminus [-Q,Q]} \rd \m \: R(\la - \m|\g) f(\m) \epc
\end{equation}
which is the complementary equation mentioned above. In particular,
\begin{equation} \label{resolvcompl}
     R_Q (\la, \m) = R (\la - \m|\g) +
        \int_{{\mathbb R} \setminus [-Q,Q]} \rd \n \: R(\la - \n|\g) R_Q (\n, \m) \epp
\end{equation}
\begin{lemma}
Solutions of (\ref{fredint}), for which $f_0$ is a uniformly
bounded continuous function on~${\mathbb R}$, can
be represented by means of the resolvent kernel in two different
ways,
\begin{subequations}
\begin{align} \label{fredresform1}
     f(\la) & = f_0 (\la) - \int_{-Q}^Q \rd \m \: R_Q (\la, \m) f_0 (\m) \\
            & = f_\infty (\la) + 
                \int_{{\mathbb R} \setminus [-Q,Q]} \rd \m \:
		   R_Q (\la, \m) f_\infty (\m) \epp \label{fredresform2}
\end{align}
\end{subequations}
\end{lemma}
\begin{proof}
For the first equation we multiply
\begin{equation}
     R_Q (\la, \m) + \int_{-Q}^Q \rd \n \: R_Q (\la, \n) K(\n - \m|\g) =
        K(\la - \m|\g)
\end{equation}
by $f(\m)$ and integrate over $\m$. Similarly, we multiply
(\ref{fredint}), with $\la$ replaced by $\m$, by $R_Q (\la, \m)$
and integrate over $\m$. It follows that
\begin{equation}
     \int_{-Q}^Q \rd \m \: K(\la - \m|\g) f(\m) =
        \int_{-Q}^Q \rd \m \: R_Q (\la, \m) f_0 (\m) \epp
\end{equation}
When reinserted into (\ref{fredint}), this proves (\ref{fredresform1}).
In order to prove (\ref{fredresform2}) apply a similar argument to
(\ref{fredintcompl}), (\ref{resolvcompl}).
\end{proof}
\begin{lemma} \label{lem:propsrq}
Bounds on $R_Q$ \cite{DGK14b}. Let $0 < \g < \p/2$. Then
\begin{enumerate}
\item
\begin{equation} \label{rqr}
     R_Q (\la, \m) > R(\la - \m|\g)
\end{equation}
uniformly in $(\la, \m) \in {\mathbb R}^2$.
\item
\begin{equation} \label{rqrqlower}
     R_Q (\la, \m) - R_Q (\la, - \m) > R(\la - \m|\g) - R(\la + \m|\g) > 0
\end{equation}
for all $\la, \m > 0$.
\end{enumerate}
\end{lemma}
\begin{proof}
(i) follows from (\ref{resolvcompl}) and the fact that $R(\la|\g) > 0$
for all $\la \in {\mathbb R}$, since all terms in the iterative
(Neumann series) solution are positive.

(ii) Using (\ref{resolvcompl}) and Lemma~\ref{lem:reskerqfinite}
we obtain
\begin{multline} \label{diffrqinteqn}
     R_Q (\la, \m) - R_Q (\la, - \m) = R(\la - \m|\g) - R(\la + \m|\g) \\
        + \int_Q^\infty \rd \n \: \bigl(R(\la - \n|\g) - R(\la + \n|\g)\bigr)
             \bigl(R_Q (\n, \m) - R_Q (\n, - \m)\bigr) \epp
\end{multline}
Since $R $ is even, $R(\la - \m|\g) = R(|\la - \m||\g)$. Since
$|\la - \m| < \la + \m$ for $\la, \m \in {\mathbb R}_+$, $R(\cdot|\g)$
being decreasing on ${\mathbb R}_+$ then implies that $R(\la - \m|\g)
- R(\la + \m) > 0$ for $\la, \m \in {\mathbb R}_+$. It follows that
the driving term of the integral equation (\ref{diffrqinteqn}) and
all its iterations are positive, which entails the claim.
\end{proof}

\section{Proofs}
\subsection{Proof of Theorem~\ref{th:fpoints}}
(i) The continuity in $Q$ follows from the continuity of
$R_Q$ in $Q$ that was established above. The evenness in $\la$
follows, since $\e_0$ is even.

(ii) The lower bound follows from (\ref{fredresform1}) with
$f_0 = \e_0$ and the fact that $R_Q (\la, \m) > 0$ for all
$\la, \m \in {\mathbb R}$ and $\e_0 (\la) < 0$ for all
$\la \in [- Q_0,Q_0]$. For the upper bound we introduce the
dressed charge function $Z(\la|Q)$, which is the solution of
(\ref{fredint}) with driving term $f_0 (\la) = 1$, and the
root density $\r (\la|Q)$, the solution of (\ref{fredint})
with $f_0 (\la) = \r_0 (\la) = K(\la|\g/2)$. Then
\begin{equation} \label{epszrho}
     \e(\la|Q) = h Z(\la|Q) - 4 \p J \sin(\g) \r (\la|Q) \epp
\end{equation}
For the dressed charge function we have the upper bound
\begin{equation} \label{zresform}
     Z(\la|Q) = 1 - \int_{-Q}^Q \rd \m \: R_Q (\la, \m) < 1 \epc
\end{equation}
since $R_Q (\la, \m) > 0$, while for the root density
\begin{equation} \label{rhoresform}
     \r(\la|Q) = \r_\infty (\la) +
        \int_{{\mathbb R} \setminus [-Q,Q]} \rd \m \:
	R_Q (\la, \m) \r_\infty (\m) > \r_\infty (\la) \epc
\end{equation}
since $\r_\infty (\la) > 0$ as well. Thus,
\begin{equation}
     \e(\la|Q) < h - 4 \p J \sin(\g) \r_\infty (\la) = \e_u (\la) \epp
\end{equation}

(iii) We take the derivative of the `resolvent form'
(\ref{fredintcompl}) of the integral equation for $\e(\cdot|Q)$,
use partial integration and the fact that $\e(\cdot|Q)$ is even.
Then
\begin{align} \label{epsprime}
     \e' (\la|Q) & = \e(Q|Q) \bigl(R(\la - Q|\g) - R(\la + Q|\g)\bigr)
        \notag \\ & \mspace{90.mu}
        + \e_\infty' (\la)
	+ \int_{{\mathbb R} \setminus [-Q,Q]} \rd \m \: R(\la - \m|\g) \e' (\m|Q)
	\notag \\[1ex]
	& = \e(Q|Q) \bigl(R_Q (\la,Q) - R_Q (\la,-Q)\bigr)
        \notag \\ & \mspace{90.mu}
        + \e_\infty' (\la)
	+ \int_Q^\infty \rd \m \: \bigl(R_Q (\la,\m) - R_Q (\la,-\m) \bigr)
	  \e_\infty' (\m) \epp
\end{align}
On the other hand
\begin{equation}
     \6_Q \e(\la|Q) = - \e(Q|Q) \bigl(R_Q (\la,Q) + R_Q (\la,-Q)\bigr) \epp
\end{equation}
Combining the latter two equations we obtain
\begin{multline} \label{totalqderivative}
     \frac{\rd \e(Q|Q)}{\rd Q} = - 2 \e(Q|Q) R_Q (Q,-Q) \\
        + \e_\infty' (Q)
	+ \int_Q^\infty \rd \m \: \bigl(R_Q (Q,\m) - R_Q (Q,-\m) \bigr)
	  \e_\infty' (\m) \epp
\end{multline}
Now $\e_\infty' (\la) > 0$ for $\la > 0$ and the bracket under
the integral is positive because of~(\ref{rqrqlower}). Thus,
$\e (Q|Q) = 0 \ \then \frac{\rd \e (Q|Q)}{\rd Q} > 0$, meaning that
every zero of $Q \mapsto \e(Q|Q)$ belongs to an open set on which
the function is increasing. Then, by its continuity on~${\mathbb R}$,
the function $Q \mapsto \e (Q|Q)$ has at most one zero. But $\e (0|0)
= \e_0 (0) = h - h_c$ and $\lim_{Q \rightarrow \infty} \e (Q|Q) =
\lim_{\la \rightarrow \infty} \e_\infty (\la) = \frac{h}{2 (1 - \g/\p)} > 0$,
implying that $Q \mapsto \e (Q|Q)$ has a unique positive zero $Q_F$ if
and only if $0 < h < h_c$.

(iv) The bounds $Q_F < Q_0$ and $Q_F > Q_u$, if $Q_u > 0$ exists,
follow from (\ref{epslowup}) and the monotonicity of $\e_0$ and
$\e_u$.

(v) The smoothness of $h \mapsto Q_F$ is consequence of the
implicit function theorem. $\frac{\rd Q_F}{\rd h}$ can be directly
calculated by implicit differentiation and the use of (\ref{epsprime}),
(\ref{totalqderivative}).
\begin{equation}
     \frac{\rd Q_F}{\rd h} = - \frac{Z(Q_F|Q_F)}{\e' (Q_F)} < 0 \epc
\end{equation}
since $Z(Q_F|Q_F) > 0$ and $\e' (Q_F) > 0$ (the latter follows from
(\ref{epsprime}), for the former one has to consider the resolvent
form of the integral equation for $Z(\la|Q)$). The limits in (v)
follow from (\ref{qfbounds}).

\subsection{Proof of Theorem \ref{th:main}}
Recall that we denote $\e = \e (\cdot|Q_F)$. Throughout this proof
we shall frequently use the notation $\la = x + \i y$ with
$x, y \in {\mathbb R}$.
\subsubsection*{Proof of (i)} Since the integral equation for $\e$ is linear, we have
\begin{align}
     \Re \e(x + \i y) & = \Re \e_0 (x + \i y)
        - \int_{- Q_F}^{Q_F} \rd \m \: \Re \bigl(K(x - \m + \i y|\g)\bigr) \e(\m)
	\notag \\[1ex] & =
	h - 2 \p J \sin(\g) \bigl(K(x|\g/2 - y) + K(x|\g/2 + y)\bigr)
	\notag \\ & \mspace{72.mu}
        - \int_{- Q_F}^{Q_F} \rd \m \:
	  \2 \bigl(K(x - \m|\g - y) + K(x - \m|\g +y)\bigr) \e(\m) \epp
\end{align}
Here we have used (\ref{rekkk}) in the second equation. The
expression on the right hand side is obviously even in $y$.
Its evenness in $x$ follows, since $\e (\m)$ is even for
$\m \in {\mathbb R}$ and since $K(\la|\g)$ is an even function
of $\la$.

\subsubsection*{Proof of (ii)}
The proof of (ii) relies on the fact that, provided one
replaces the functions in (\ref{rqr}), (\ref{rqrqlower}) by
their real parts, Lemma~\ref{lem:propsrq} can be extended
for $\la$ in the strip $|\Im \la| < \g/2$, which is
essentially due to the fact that (\ref{rconv}) holds in
that strip. We start with the elementary formulae
\begin{subequations}
\label{reoneoverch}
\begin{align} \label{rechinv}
     \Re \biggl(\frac{1}{\ch(\la \p/\g)}\biggr) & =
        \frac{\ch(x \p/\g) \cos(y \p/\g)}{\sh^2(x \p/\g) + \cos^2(y \p/\g)} \epc \\[1ex]
     \6_x \Re \biggl(\frac{1}{\ch(\la \p/\g)}\biggr) & = - \frac{\p}{\g}
        \frac{\sh(x \p/\g) \cos(y \p/\g) \bigl(\ch^2(x \p/\g) + \sin^2(y \p/\g)\bigr)}
	     {\bigl(\sh^2(x \p/\g) + \cos^2(y \p/\g)\bigr)^2} \epc \label{derrechinv}
\end{align}
\end{subequations}
which show that $\Re \bigl(1/\ch(\la \p/\g)\bigr)$ as a function of
$x = \Re \la$ is even and positive on ${\mathbb R}$ and monotonically
decreasing on ${\mathbb R}_+$, if $y = \Im \la \in (-\g/2, \g/2)$.

Taking the real part of (\ref{rconv}) we conclude with (\ref{rechinv})
that $\Re R(\la|\g) > 0$ for all $x \in {\mathbb R}$, if $y \in
(-\g/2, \g/2)$. Similarly, taking the real part of (\ref{rprimereal}),
using (\ref{reoneoverch}) and the fact that $K'(\la|\g') < 0$ for
$\la \in {\mathbb R}_+$, we conclude that $\6_x \Re R(\la|\g) < 0$
for all $x \in {\mathbb R}_+$, if $y \in (-\g/2, \g/2)$.

Taking the real part of (\ref{resolvcompl}) and using that 
$\Re R(\la|\g) > 0$ we conclude that
\begin{equation} \label{rerqgreater}
     \Re R_Q (\la, \m) > \Re R(\la - \m|\g) > 0
\end{equation}
for all $x \in {\mathbb R}$, $\m \in {\mathbb R}$, if $y \in
(-\g/2, \g/2)$. Similarly, taking the real part of (\ref{diffrqinteqn})
and using that $\Re R(\la|\g)$ is even, positive and monotonically
decreasing for $x \in {\mathbb R}_+$ we obtain the inequality
\begin{equation} \label{rerqminusrq}
     \Re \bigl(R_Q (\la,\m) - R_Q (\la, -\m)\bigr) >
        \Re \bigl(R (\la - \m|\g) - R(\la + \m|\g)\bigr) > 0
\end{equation}
for all $x > 0$, $\m > 0$, if $y \in (-\g/2, \g/2)$.

Setting $Q = Q_F$ in (\ref{epsprime}) and taking the real
part implies that
\begin{equation} \label{reepsprime}
     \6_x \Re \e(\la) = \Re \e_\infty' (\la) +
        \int_{Q_F}^\infty \rd \m \:
	   \Re \bigl(R_{Q_F} (\la, \m) - R_{Q_F} (\la, -\m)\bigr)
	   \e_\infty' (\m) \epp
\end{equation}
Here $\Re \e_\infty' (\la) > 0$ due to (\ref{derrechinv}), and the
integral is positive as well, because of (\ref{rerqminusrq}) and
since $\e_\infty' (\m) > 0$ for all $\m \in {\mathbb R}_+$. Thus,
we have shown that $\Re \e (\la)$ is monotonically increasing as
a function of $x = \Re \la$ for all $x > 0$ and all $y \in (-\g/2, \g/2)$.

The facts that $\e$ is bounded on $[-Q_F,Q_F]$, that $\lim_{\la
\rightarrow + \infty} \Re K(\la - \m|\g) = 0$, uniformly for all
$\m \in [-Q_F,Q_F]$, and that $\e$ satisfies (\ref{fredint}) imply
that
\begin{equation}
     \lim_{x \rightarrow \infty} \Re \e (\la) =
     \lim_{x \rightarrow \infty} \Re \e_0 (\la) = h \epp
\end{equation}
In order to understand the behaviour of $\Re \e(\la)$ at $x = 0$,
we consider the second derivative. Starting from the resolvent
form of the integral equation for $\e$ we obtain
\begin{equation} \label{reepspp}
     \6_x^2 \Re \e(\la)\Bigr|_{\la = \i y} = \Re \e_\infty'' (\i y) -
        \int_{Q_F}^\infty \rd \m \:
	   \Re \bigl(R' (\m + \i y|\g) + R' (\m - \i y|\g)\bigr)
	   \e' (\m) \epp
\end{equation}
Here the first term under the integral is negative (as $\6_x
\Re R(\la|\g) < 0$ for $x \in {\mathbb R}_+$ and $y \in (-\g/2,
\g/2)$ (see below (\ref{reoneoverch}))) and the second term, $\e' (\m)$,
is positive. We further have the explicit result
\begin{equation}
     \e_\infty'' (\i y) = 2J \sin(\g) \biggl(\frac{\p}{\g}\biggr)^3
        \frac{1 + \sin^2(y\p/\g)}{\cos^3(y\p/\g)} > 0
\end{equation}
for all $y \in (-\g/2,\g/2)$. Thus, altogether $\6_x^2 \Re
\e(\la)\bigr|_{x = 0} > 0$ if $y \in (-\g/2,\g/2)$. Since
$\Re \e(\la)$ is harmonic, this ensures that $\6_y^2 \Re \e(\i y) < 0$.
But $\e' (0) = 0$, since $\e$ is even, and therefore
$\6_y \Re \e(\i y)\bigr|_{y = 0} = 0$. It follows that
$\6_y \Re \e(\i y) < 0$ on $(0,\g/2)$. Then, since $\e$ is even,
$y \mapsto \Re \e(\i y)$ has a unique maximum at $y = 0$ on
$(-\g/2,\g/2)$, and $\Re \e(\i y) < \e (0) < 0$ for all
$y \in (-\g/2,\g/2)$.

It follows that $x \mapsto \Re \e(x + \i y)$ has a unique positive
zero for every $y \in (0,\g/2)$. This defines a function $(0,\g/2)
\rightarrow {\mathbb R}_+$, $y \mapsto x(y)$ which is smooth due to
the implicit function theorem.

\subsubsection*{Proof of (iii)}
$x(0) = Q_F$ by definition of the Fermi rapidity. The
behaviour of the curve close to the pole of $\e$ at $\i \g/2$
follows from a perturbative analysis of the integral equation
for $\e$ in its resolvent form (\ref{fredintcompl}).

\subsubsection*{Proof of (iv)}
The lower bound follows, since
\begin{equation}
     \Re \e(\la) = \Re \e_0 (\la)
         - \int_{- Q_F}^{Q_F} \rd \m \:
	   \Re \bigl(R_{Q_F} (\la, \m)\bigr) \e_0 (\m) \epc
\end{equation}
where $\e_0 (\m) < 0$ for $\m \in [-Q_F,Q_F]$ due to (\ref{qfbounds})
and where $\Re \bigl(R_Q (\la, \m)\bigr) > 0$ according to
(\ref{rerqgreater}). For the upper bound we set $Q = Q_F$ in
(\ref{epszrho}) and take the real part,
\begin{equation} \label{reepszrho}
     \Re \e(\la) = h \Re Z(\la|Q_F) - 4 \p J \sin(\g) \Re \r (\la|Q_F) \epp
\end{equation}
Here $\Re Z(\la|Q_F) < 1$ which follows from the first equation
in (\ref{zresform}) and from (\ref{rerqgreater}), and $\Re \r(\la|Q_F)
> \Re \r_\infty (\la)$ which is a consequence of (\ref{rhoresform})
and (\ref{rerqgreater}).

\subsubsection*{Proof of (\ref{reeps4})}
The fact that $\Re \e (\la) > 0$ for all $\la \in S_\g (Q_F)$
with $|\Im \la| > \g/2$ follows from the estimates (\ref{reepsgreater})
which we shall now show one by one.

We start with (\ref{reeps4}) and assume for a while that $\frac{\p}{2}
- \2 \bigl(\frac{\p}{2} - \g\bigr) < y < \frac{\p}{2}$. In a first
step we derive an appropriate integral representation of the dressed
energy in this strip. For this purpose we start from the defining
integral equation, (\ref{fredint}) with $f_0 = \e_0$ and $Q = Q_F$,
and deform the contour as sketched in Figure~\ref{fig:deformed_contour}.
\begin{figure}
\begin{center}
\includegraphics[width=.92\textwidth]{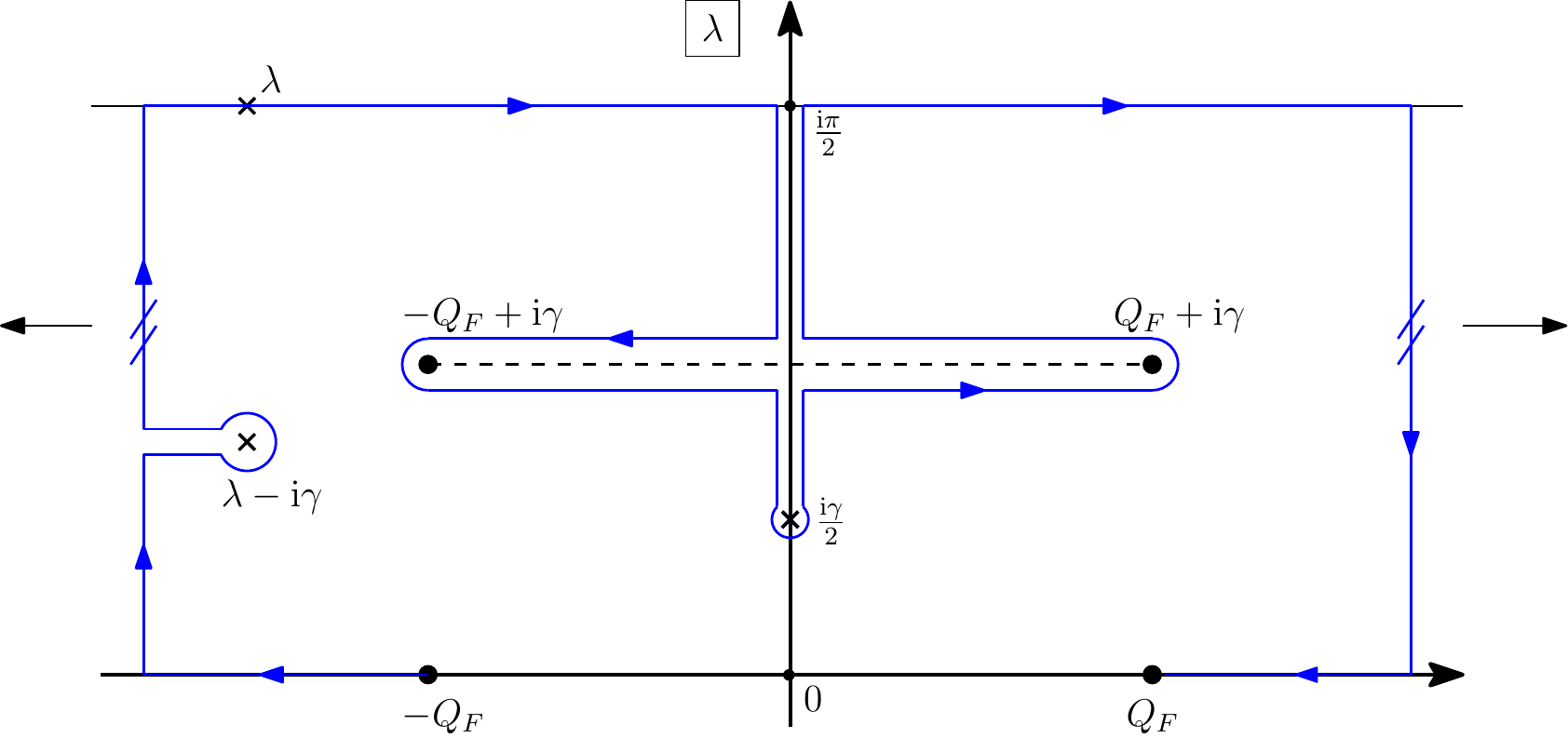}
\end{center}
\caption{\label{fig:deformed_contour} We deform the original integration
contour, which is a straight line from $-Q_F$ to $Q_F$ to the sketched
contour and move the left and the right parts to minus and plus infinity.
}
\end{figure}
Directly from the defining integral equation we can read off the
following properties of the dressed energy in the strip $0 < \Im \la <
\frac{\p}{2}$.
\renewcommand{\labelenumi}{(\alph{enumi})}
\begin{enumerate}
\item
$\e$ has a simple pole at $\la = \i \frac{\g}{2}$ with residue
\begin{equation}
     \res_{\la = \i \frac{\g}{2}} \e (\la) = 2 \i J \sin (\g) \epp
\end{equation}
\item
$\e$ has a jump discontinuity across the cut at $[-Q_F,Q_F] + \i \g$,
where
\begin{equation}
     \e_+ (\la) - \e_- (\la) = \e(\la - \i \g) \epp
\end{equation}
\item
\begin{equation} \label{epsasy}
     \lim_{\Re \la \rightarrow \pm \infty} \e (\la) = h \epp
\end{equation}
\end{enumerate}
\renewcommand{\labelenumi}{(\roman{enumi})}
We now first of all choose $\la$ such that $\Im \la = \frac{\p}{2}$.
Evaluating the integral that occurs in the integral equation for
$\e$ along the original contour and along the deformed contour and
using the above properties and the properties of the kernel we
obtain the identity
\begin{align} \label{defcontourid}
     \int_{-Q_F}^{Q_F} & \rd \m \: K(\la - \m|\g) \e (\m) \notag \\
       & = - \e (\la - \i \g) - 4 \p J \sin(\g) K(\la - \i \g/2|\g)
           - \int_{-Q_F}^{Q_F} \rd \m \: K(\la - \m - \i \g|\g) \e(\m)
	     \notag \\ & \mspace{36.mu}
           + \int_{{\mathbb R} + \i \frac{\p}{2}} \rd \m \:
	      K(\la - \m|\g) \e(\m)
           - \int_{{\mathbb R} \setminus [-Q_F,Q_F]} \rd \m \:
	      K(\la - \m|\g) \e(\m) \notag \\[1ex]
       & = - \e_0 (\la - \i \g) - 4 \p J \sin(\g) K(\la - \i \g/2|\g)
	     \notag \\ & \mspace{36.mu}
           + \int_{{\mathbb R} + \i \frac{\p}{2}} \rd \m \:
	      K(\la - \m|\g) \e(\m)
           - \int_{{\mathbb R} \setminus [-Q_F,Q_F]} \rd \m \:
	      K(\la - \m|\g) \e(\m) \epp
\end{align}
We insert this into the defining integral equation for $\e$ and
combine the explicit terms on the right hand side of
(\ref{defcontourid}) with the driving term. It follows that
\begin{equation} \label{ompreform}
     \e (\la) = 2h 
           - \int_{{\mathbb R} + \i \frac{\p}{2}} \rd \m \:
	      K(\la - \m|\g) \e(\m)
           + \int_{{\mathbb R} \setminus [-Q_F,Q_F]} \rd \m \:
	      K(\la - \m|\g) \e(\m) \epp
\end{equation}
We set $\la = z + \i \frac{\p}{2}$ and
\begin{equation}
     \om (z) = \e(z + \i \p/2) \epp
\end{equation}
Then, using the $\i \p$-periodicity of the kernel,
(\ref{ompreform}) turns into
\begin{equation}
     \om (z) = 2h 
        - \int_{{\mathbb R} \setminus [-Q_F,Q_F]} \rd w \:
	     K(z - w|\p/2 - \g) \e(w)
	- \int_{-\infty}^\infty \rd w \: K(z - w|\g) \om (w) \epp
\end{equation}

This equation can be solved for $\om$ by employing Fourier
transformation and the convolution theorem. For $\la \in
{\mathbb R}$ let
\begin{equation}
     \e_c (\la) =
        \begin{cases}
	   \e (\la) & \text{if $\la \in {\mathbb R} \setminus [-Q_F,Q_F]$} \\
	   0 & \text{else.}
        \end{cases}
\end{equation}
Then
\begin{equation}
     {\cal F} [\om] (k) =
        \frac{2 \p h \, \de(k)}{1 - \frac{\g}{\p}}
	- {\cal F} [D] (k) {\cal F} [\e_c] (k) \epc
\end{equation}
where
\begin{equation}
     {\cal F} [D] (k) = \frac{\sh\bigl(\frac{\g k}{2}\bigr)}
                             {\sh\bigl((\p - \g)\frac{k}{2}\bigr)} \epp
\end{equation}
It follows that
\begin{equation}
     D(z) = \frac{K \Bigl(\frac{z}{1 - \frac{\g}{\p}}\Big|\frac{\p}{2} - \g'\Bigr)}
                 {1 - \frac{\g}{\p}}
\end{equation}
with $\g'$ as defined in (\ref{defgprime}). Hence, $\om$ has the
representation
\begin{equation} \label{omrep}
     \om(z) = \frac{h}{1 - \frac{\g}{\p}} -
              \frac{1}{1 - \frac{\g}{\p}}
	      \int_{{\mathbb R} \setminus [-Q_F,Q_F]} \rd w \:
	         K \biggl(\frac{z - w}{1 - \frac{\g}{\p}}\bigg|\frac{\p}{2} - \g'\biggr)
		 \e(w) \epp
\end{equation}
Recall that $\g \mapsto \g'$ is a monotonically increasing bijection
of the interval $(0,\p/2)$. Hence, $\g \mapsto \frac{\p}{2} - \g'$
is a monotonically decreasing function that maps $(0,\p/2)$ onto
itself. Further notice that the kernel in (\ref{omrep}) as a function
of $z - w$ has simple poles at $\pm \i \bigl(\frac{\p}{2} - \g\bigr)
\mod \i(\p - \g)$.

We shall use (\ref{omrep}) to establish the lower bounds (\ref{reeps4})
and (\ref{reeps3}). Let us begin with (\ref{reeps4}). Since $K$ is
harder to estimate than $R$ we use the integral equation (\ref{intr})
in order to replace the kernel function $K$ on the right hand side of
(\ref{omrep}). Setting $x = \Re z = \Re \la$, $b = \Im z = \Im \la - \p/2
= y - \p/2$ and taking the real part and the $x$-derivative of (\ref{omrep})
we arrive after a few elementary manipulations at
\begin{align} \label{deromrep}
     \6_x \Re \om (z) = &
        - \frac{1}{1 - \frac{\g}{\p}}
	  \int_{Q_F}^\infty \rd w \: \Re \biggl\{
	  R \biggl(\frac{z - w}{1 - \frac{\g}{\p}}\bigg|\frac{\p}{2} - \g'\biggr)
	  - R \biggl(\frac{z + w}{1 - \frac{\g}{\p}}\bigg|\frac{\p}{2} - \g'\biggr)
	  \biggr\} \e' (w) \notag \\[1ex] & \mspace{-90.mu} -
          \frac{1}{\bigl(1 - \frac{\g}{\p}\bigr)^2}
	  \int_{0}^\infty \rd w \: \Re \biggl\{
	  R \biggl(\frac{z - w}{1 - \frac{\g}{\p}}\bigg|\frac{\p}{2} - \g'\biggr)
	  - R \biggl(\frac{z + w}{1 - \frac{\g}{\p}}\bigg|\frac{\p}{2} - \g'\biggr)
	  \biggr\} \notag \\ & \mspace{-18.mu} \times
	  \int_{Q_F}^\infty \rd u \: \Re \biggl\{
	  K \biggl(\frac{w - u}{1 - \frac{\g}{\p}}\bigg|\frac{\p}{2} - \g'\biggr)
	  - K \biggl(\frac{w + u}{1 - \frac{\g}{\p}}\bigg|\frac{\p}{2} - \g'\biggr)
	  \biggr\} \e' (u) \epp
\end{align}
Now (\ref{rerqminusrq}) implies that
\begin{equation}
   \Re \biggl\{
      R \biggl(\frac{z - w}{1 - \frac{\g}{\p}}\bigg|\frac{\p}{2} - \g'\biggr)
      - R \biggl(\frac{z + w}{1 - \frac{\g}{\p}}\bigg|\frac{\p}{2} - \g'\biggr)
      \biggr\} > 0
\end{equation}
for all $x = \Re z > 0$, $w > 0$ if
\begin{equation} \label{yrangeom4}
     |b| < \biggl(1 - \frac{\g}{\p}\biggr) \2 \biggl(\frac{\p}{2} - \g'\biggr)
         = \2 \biggl(\frac{\p}{2} - \g \biggr) \epp
\end{equation}
Since $\e' (w) > 0$ for $w > 0$ and since the same is true for
the difference of the kernel functions $K$ in the second
integral on the right hand side of (\ref{deromrep}) we conclude that
$\6_x \Re \om (z) < 0$, meaning that $x \mapsto \Re \om (x + \i b)$
is monotonically decreasing on ${\mathbb R}_+$ if $b$ satisfies
(\ref{yrangeom4}). Thus, $x \mapsto \Re \e (x + \i y)$ is monotonically
decreasing on ${\mathbb R}_+$ for $\frac{\p}{2} -
\2 \bigl(\frac{\p}{2} - \g\bigr) < y < \frac{\p}{2}$. Combining this
knowledge with the asymptotic formula (\ref{epsasy}) we have established
(\ref{reeps4}).

\subsubsection*{Proof of (\ref{reeps3})}
We proceed with (\ref{reeps3}). For the proof we consider (\ref{omrep})
with
\begin{equation} \label{rangeom3}
     \g - \frac{\p}{2} < b = \Im z < \2 \biggl(\g - \frac{\p}{2}\biggr) < 0 \epp
\end{equation}
Equation (\ref{rekkk}) implies that
\begin{equation}
   \Re K \biggl(\frac{z - w}{1 - \frac{\g}{\p}}\bigg|\frac{\p}{2} - \g'\biggr)
      = \2 \biggl\{
        K \biggl(\frac{x - w}{1 - \frac{\g}{\p}}\bigg| \g_+ \biggr)
        + K \biggl(\frac{x - w}{1 - \frac{\g}{\p}}\bigg| \g_- \biggr)\biggr\} \epc
\end{equation}
where $x = \Re z$ and
\begin{equation}
     \g_\pm = \frac{\p}{2} - \g' \pm \frac{b}{1 - \frac{\g}{\p}} \epp
\end{equation}
The inequality (\ref{rangeom3}) implies that
\begin{equation}
     0 < \g_+ < \frac{\p}{4} \epc \qd \g_+ < \g_- < \p \epp
\end{equation}
Hence, we have to distinguish two cases, $\g_- < \frac{\p}{2}$ or
$\g_- > \frac{\p}{2}$.

Taking the real part and the $x$-derivative of (\ref{omrep}) we see that
\begin{equation}
     \6_x \Re \om (z) =
        - \frac{1}{1 - \frac{\g}{\p}}
	\int_{Q_F}^\infty \rd w \: \sum_{\s = \pm}
           \biggl\{
           K \biggl(\frac{x - w}{1 - \frac{\g}{\p}}\bigg| \g_\s \biggr)
           - K \biggl(\frac{x + w}{1 - \frac{\g}{\p}}\bigg| \g_\s \biggr)\biggr\}
	   \frac{\e' (w)}{2} \epp
\end{equation}
If $\g_- < \frac{\p}{2}$, then the summands for $\s = +$ and for $\s = -$
are both positive. Since moreover $\e' (w) > 0$ for $w > 0$, we conclude
that $x \mapsto \Re \om(x + \i b)$ is monotonically decreasing on
${\mathbb R}_+$. Thus, (\ref{epsasy}) implies that $\Re \om(x + \i b) > h$
in this case.

On the other hand, if $\g_- > \frac{\p}{2}$, then $K(x/(1 - \g/\p)|\g_-) < 0$
for all $x \in {\mathbb R}$ and
\begin{equation}
     \Re \om (z) >
        \frac{h}{1 - \frac{\g}{\p}}
	   - \frac{1}{1 - \frac{\g}{\p}}
	     \int_{{\mathbb R} \setminus [-Q_F,Q_F]} \rd w \:
	     K \biggl(\frac{x - w}{1 - \frac{\g}{\p}}\bigg|\g_+ \biggr)
	     \frac{\e(w)}{2} = F(x) \epp
\end{equation}
As above we can conclude that $F'(x) < 0$ for all $x > 0$. For
the asymptotic behaviour of this function we obtain
\begin{multline}
     \lim_{x \rightarrow \infty} F(x) = 
        \frac{h}{1 - \frac{\g}{\p}}
	   - \frac{h}{2} \int_{-\infty}^\infty \rd w \: K(w|\g_+) \\ =
        \frac{h}{1 - \frac{\g}{\p}}
	- \frac{h}{\p} \biggl(\g' - \frac{b}{1 - \frac{\g}{\p}}\biggr)
	> \frac{h}{2} \frac{1 + \frac{\g}{\p}}{1 - \frac{\g}{\p}} > \frac{h}{2} \epp
\end{multline}
Thus, $\Re \om (z) > \frac{h}{2}$, whenever (\ref{rangeom3}) is
satisfied, or $\Re \e (x + \i y) > \frac{h}{2}$ for all $x \in
{\mathbb R}$ if $\g < y < \frac{\p}{2} - \2 \bigl(\frac{\p}{2}
- \g\bigr)$, which is (\ref{reeps3}).

\subsubsection*{Proof of (\ref{reeps2})}
It remains to prove (\ref{reeps2}). For this purpose we start with
(\ref{fredintcompl}) for the dressed energy,
\begin{equation} \label{epsintcompl}
     \e (\la) = \e_\infty (\la) +
        \int_{{\mathbb R} \setminus [-Q_F,Q_F]} \rd \m \:
	R(\la - \m|\g) \e(\m) \epp
\end{equation}
First of all equation (\ref{rechinv}) with $\frac{\g}{2} < y
< \g$ implies that
\begin{equation} \label{rinfest}
     \Re \e_\infty (\la) > \frac{h}{2 \bigl(1 - \frac{\g}{\p}\bigr)} \epp
\end{equation}
In order to estimate the real part of the integral in
(\ref{epsintcompl}), we define the function
\begin{equation} \label{rconvext}
     R_I (\la|\g) =
        \frac{\p}{2 \g(\p - \g)} \int_{- \infty}^\infty \rd \m \:
	\frac{K \Bigl(\frac{\m}{1 - \frac{\g}{\p}}\Big|\frac{\g/2}{1 - \frac{\g}{\p}}\Bigr)}
	     {\ch \bigl((\la - \m) \frac{\p}{\g}\bigr)}
\end{equation}
which, for $\frac{\g}{2} < y = \Im \la < \g$ \emph{differs} from
the resolvent $R(\la|\g)$. Analytically continuing (\ref{rconv})
we rather see that
\begin{equation}
     R(\la|\g) = R_I (\la|\g)
        + \frac{1}{1 - \frac{\g}{\p}}
	   K \biggl(\frac{\la - \i \g/2}{1 - \frac{\g}{\p}}
	     \bigg|\frac{\g/2}{1 - \frac{\g}{\p}}\biggr) \epp
\end{equation}
Hence,
\begin{multline} \label{intseps2}
     \int_{{\mathbb R} \setminus [-Q_F,Q_F]} \rd \m \:
        \Re \bigl(R(\la - \m|\g)\bigr) \e(\m) =
        \int_{{\mathbb R} \setminus [-Q_F,Q_F]} \rd \m \:
           \Re \bigl(R_I (\la - \m|\g)\bigr) \e(\m) \\[1ex] \mspace{-1.mu}
        + \frac{1}{2\bigl(1 - \frac{\g}{\p}\bigr)}
          \int_{{\mathbb R} \setminus [-Q_F,Q_F]} \rd \m \:
	   \biggl\{
	   K \biggl(\frac{x - \m}{1 - \frac{\g}{\p}}
	     \bigg|\frac{\g - y}{1 - \frac{\g}{\p}}\biggr)
	   + K \biggl(\frac{x - \m}{1 - \frac{\g}{\p}}
	       \bigg|\frac{y}{1 - \frac{\g}{\p}}\biggr)\biggr\} \e (\m) \epp
\end{multline}

Now $\Re R_I (\la - \m) < 0$ for all $y = \Im \la \in (\g/2, \g)$,
$x = \Re \la, \m \in {\mathbb R}$, because of (\ref{rechinv}),
(\ref{rconvext}), while $0 < \e (\m) < h$ for all $\m \in {\mathbb R}
\setminus [- Q_F,Q_F]$. Thus,
\begin{align} \label{riintest}
     & \int_{{\mathbb R} \setminus [-Q_F,Q_F]} \rd \m \:
        \Re \bigl(R_I (\la - \m|\g)\bigr) \e(\m)
	 \notag \\[1ex] & \mspace{18.mu}
     > h \int_{{\mathbb R} \setminus [-Q_F,Q_F]} \rd \m \:
         \Re \bigl(R_I (\la - \m|\g)\bigr)
     > h \int_{- \infty}^\infty \rd \m \:
         \Re \bigl(R_I (\m - \i y|\g)\bigr)
	 \notag \\[1ex] & \mspace{18.mu}
     = \frac{h}{2 \g \bigl(1 - \frac{\g}{\p}\bigr)}
       \int_{- \infty}^\infty \rd \n \:
          K \Bigl(\frac{\n}{1 - \frac{\g}{\p}}\Big|\frac{\g/2}{1 - \frac{\g}{\p}}\Bigr)
          \Re \int_{- \infty}^\infty \rd \m \:
              \frac{1}{\ch \bigl((\m - \n - \i y) \frac{\p}{\g}\bigr)}
	 \notag \\[1ex] & \mspace{18.mu}
     = - \frac{h}{2} \frac{1 - \frac{2 \g}{\p}}{1 - \frac{\g}{\p}} \epp
\end{align}
Here we have used the residue theorem and equation (\ref{kfourier})
to evaluate the integrals in the last equation.

In order to estimate the second integral in (\ref{intseps2}) we
note that $\frac{\g}{2} < y < \g$ implies that
\begin{equation}
     0 < \frac{\g - y}{1 - \frac{\g}{\p}} < \g' < \frac{y}{1 - \frac{\g}{\p}} < 2 \g' \epp
\end{equation}
Recalling that $0 < \g' < \frac{\p}{2}$ we see that the
contribution of the first kernel to the second integral in
(\ref{intseps2}) is always positive. Hence,
\begin{align} \label{kkintest}
     & \frac{1}{2\bigl(1 - \frac{\g}{\p}\bigr)}
        \int_{{\mathbb R} \setminus [-Q_F,Q_F]} \rd \m \:
	   \biggl\{
	   K \biggl(\frac{x - \m}{1 - \frac{\g}{\p}}
	     \bigg|\frac{\g - y}{1 - \frac{\g}{\p}}\biggr)
	   + K \biggl(\frac{x - \m}{1 - \frac{\g}{\p}}
	       \bigg|\frac{y}{1 - \frac{\g}{\p}}\biggr)\biggr\} \e (\m)
	\notag \\[1ex] & \mspace{18.mu}
        > \frac{1}{2\bigl(1 - \frac{\g}{\p}\bigr)}
          \int_{{\mathbb R} \setminus [-Q_F,Q_F]} \rd \m \:
	     K \biggl(\frac{x - \m}{1 - \frac{\g}{\p}}
	       \bigg|\frac{y}{1 - \frac{\g}{\p}}\biggr) \e(\m)
	\notag \\[1ex] & \mspace{18.mu} >
	\begin{cases}
	   0 & \text{if $\frac{y}{1 - \frac{\g}{\p}} < \frac{\p}{2}$} \\[1ex]
	   \frac{h}{2} \frac{1 - \frac{3 \g}{\p}}{1 - \frac{\g}{\p}}
	   & \text{if $\frac{y}{1 - \frac{\g}{\p}} > \frac{\p}{2}$.}
	\end{cases}
\end{align}
For the second case in the last line we have estimated the
integral by replacing $\e (\m)$ by $h$ and the range
of integration by the real axis as in (\ref{riintest}).
Combining the estimates (\ref{rinfest}), (\ref{riintest}) and
(\ref{kkintest}) we arrive at the conclusion that
\begin{equation}
     \Re \e (\la) >
        \begin{cases}
	   \frac{h \g}{\p - \g} & \text{if $y < \frac{\p}{2} - \frac{\g}{2}$} \\[1ex]
	   \frac{h}{2} & \text{if $y > \frac{\p}{2} - \frac{\g}{2}$}
	\end{cases}
\end{equation}
which entails the claim (\ref{reeps2}).

\subsubsection*{Proof of (vi) and (vii)}
(vi) is a consequence of an analogous property of the kernel
function $K(\cdot|\g)$. In order to show (vii) we introduce the
notation $u = \Re \e$, $v = \Im \e$, $x = \Re \la$, $y = \Im \la$
and consider $\la$ as $\la(y) = x(y) + \i y$. The curve $u(\la) = 0$
is located in the strip $|y| < \g/2$. In this strip $u_x > 0$
according to (ii). By implicit differentiation
\begin{equation}
     \frac{\rd x}{\rd y} = - \frac{u_y}{u_x} \epp
\end{equation}
It follows that
\begin{equation}
     \frac{\rd v}{\rd y} = v_x \frac{\rd x}{\rd y} + v_y
        = - \frac{v_x u_y}{u_x} + v_y = \frac{v_x^2}{u_x} + u_x > 0 \epp
\end{equation}
Here we have used the Cauchy-Riemann equations in the third equation.
Equation (\ref{leadim}) is obtained by inserting (\ref{leadx}) into
the leading term of the Laurent expansion of $\e$ obtained, for
instance, from the integral equation (\ref{fredint}) with
$f_0 = \e_0$, $Q = Q_F$.

\section{Conclusions}
We have studied some of the properties of the dressed energy $\e$
of the XXZ chain in the complex plane. In particular, we have
obtained a clear picture of where $\Re \e$ is positive and where
it is negative. Both regions are separated by the smooth simple
and closed curve $\Re \e = 0$, which is reflection symmetric with
respect to real and imaginary axis, which goes through two Fermi
points $\pm Q_F$ on the real axis and through the points $\pm \i \g/2$.
Moreover, this curve is entirely located inside the strip $\Im \la
\le \g/2$. In the right half plane $\Im \e$ is monotonically
increasing in the direction of increasing imaginary part and is
diverging at $\pm \i \g/2$. These properties are essential for
a future rigorous characterization of the auxiliary functions
that determine the sets of Bethe roots and the eigenvalues of the
quantum transfer matrix of the model in the zero-temperature limit.

{\bf Acknowledgments.}
The authors would like to thank Junji Suzuki for a critical reading
of the manuscript and for his valuable comments. SF and FG
acknowledge financial support by the DFG in the framework
of the research unit FOR 2316. The work of KKK is supported by the
CNRS and by the ‘Projet international de coop\'eration scientifique
No. PICS07877’: \textit{Fonctions de corr\'elations dynamiques
dans la cha{\^\nodoti}ne XXZ \`a temp\'erature finie}, Allemagne,
2018-2020.

%
%
%
%


\begin{thebibliography}{10}

\bibitem{BGKS21a}
C.~Babenko, F.~G\"ohmann, K.~K. Kozlowski, and J.~Suzuki, \emph{A thermal form
  factor series for the longitudinal two-point function of the
  {H}eisenberg-{I}sing chain in the antiferromagnetic massive regime}, J. Math.
  Phys. \textbf{62} (2021), 041901.

\bibitem{Bethe31}
H.~Bethe, \emph{Zur {T}heorie der {M}etalle. {I}. {E}igenwerte und
  {E}igenfunktionen der line\-aren {A}tomkette}, Z. Phys. \textbf{71} (1931),
  205--226.

\bibitem{DGK13a}
M.~Dugave, F.~G\"ohmann, and K.~K. Kozlowski, \emph{Thermal form factors of the
  {XXZ} chain and the large-distance asymptotics of its temperature dependent
  correlation functions}, J. Stat. Mech.: Theor. Exp. (2013), P07010.

\bibitem{DGK14b}
\bysame, \emph{Functions characterizing the ground state of the {XXZ}
  spin-$1/2$ chain in the thermodynamic limit}, SIGMA \textbf{10} (2014), 043.

\bibitem{DGK14a}
\bysame, \emph{Low-temperature large-distance asymptotics of the transversal
  two-point functions of the {XXZ} chain}, J. Stat. Mech.: Theor. Exp. (2014),
  P04012.

\bibitem{DGKS15b}
M.~Dugave, F.~G\"ohmann, K.~K. Kozlowski, and J.~Suzuki, \emph{Low-temperature
  spectrum of correlation lengths of the {XXZ} chain in the antiferromagnetic
  massive regime}, J. Phys. A \textbf{48} (2015), 334001.

\bibitem{DGKS16b}
\bysame, \emph{Thermal form factor approach to the ground-state correlation
  functions of the {XXZ} chain in the antiferromagnetic massive regime}, J.
  Phys. A \textbf{49} (2016), 394001.

\bibitem{GGK00}
I.~Gohberg, S.~Goldberg, and N.~Krupnik, \emph{Traces and determinants of
  linear operators}, Operator Theory -- Advances and Applications, vol. 116,
  Birkh\"auser Verlag, Basel, 2000.

\bibitem{GKKKS17}
F.~G\"ohmann, M.~Karbach, A.~Kl\"umper, K.~K. Kozlowski, and J.~Suzuki,
  \emph{Thermal form-factor approach to dynamical correlation functions of
  integrable lattice models}, J. Stat. Mech.: Theor. Exp. (2017), 113106.

\bibitem{KKMST11b}
N.~Kitanine, K.~K. Kozlowski, J.~M. Maillet, N.~A. Slavnov, and V.~Terras,
  \emph{A form factor approach to the asymptotic behavior of correlation
  functions in critical models}, J. Stat. Mech.: Theor. Exp. (2011), P12010.

\bibitem{KMT99b}
N.~Kitanine, J.~M. Maillet, and V.~Terras, \emph{Correlation functions of the
  {XXZ} {H}eisenberg spin-$\frac{1}{2}$ chain in a magnetic field}, Nucl. Phys.
  B \textbf{567} (2000), 554.

\bibitem{Kluemper92}
A.~Kl\"umper, \emph{Free energy and correlation length of quantum chains
  related to restricted solid-on-solid lattice models}, Ann.\ Physik \textbf{1}
  (1992), 540.

\bibitem{Orbach58}
R.~Orbach, \emph{Linear antiferromagnetic chain with anisotropic coupling},
  Phys. Rev. \textbf{112} (1958), 309.

\bibitem{STF79}
E.~K. Sklyanin, L.~A. Takhtadzhyan, and L.~D. Faddeev, \emph{Quantum inverse
  problem method. {I}.}, Theor. Math. Phys. \textbf{40} (1979), 688.

\bibitem{SAW90}
J.~Suzuki, Y.~Akutsu, and M.~Wadati, \emph{A new approach to quantum spin
  chains at finite temperature}, J. Phys. Soc. Jpn. \textbf{59} (1990), 2667.

\bibitem{Suzuki85}
M.~Suzuki, \emph{Transfer-matrix method and {Monte Carlo} simulation in quantum
  spin systems}, Phys. Rev. B \textbf{31} (1985), 2957.

\bibitem{TaSu72}
M.~Takahashi and M.~Suzuki, \emph{One-dimensional anisotropic {H}eisenberg
  model at finite temperatures}, Prog. Theor. Phys. \textbf{48} (1972), 2187.

\bibitem{Walker59}
L.~R. Walker, \emph{Antiferromagnetic linear chain}, Phys. Rev. \textbf{116}
  (1959), 1089.

\bibitem{YaYa66a}
C.~N. Yang and C.~P. Yang, \emph{Ground-state energy of a {H}eisenberg-{I}sing
  lattice}, Phys. Rev. \textbf{147} (1966), 303.

\bibitem{YaYa66b}
\bysame, \emph{One-dimensional chain of anisotropic spin-spin interactions.
  {I}.\ {P}roof of {B}ethe's hypothesis for ground state in a finite system},
  Phys. Rev. \textbf{150} (1966), 321.

\bibitem{YaYa66c}
\bysame, \emph{One-dimensional chain of anisotropic spin-spin interactions.
  {II}.\ {P}roperties of the ground-state energy per lattice site for an
  infinite system}, Phys. Rev. \textbf{150} (1966), 327.

\bibitem{YaYa69}
\bysame, \emph{Thermodynamics of a one-dimensional system of {B}osons with
  repulsive delta-function interaction}, J. Math. Phys. \textbf{10} (1969),
  1115.

\end{thebibliography}

\providecommand{\bysame}{\leavevmode\hbox to3em{\hrulefill}\thinspace}
\providecommand{\MR}{\relax\ifhmode\unskip\space\fi MR }
\providecommand{\MRhref}[2]{%
  \href{http://www.ams.org/mathscinet-getitem?mr=#1}{#2}
}
\providecommand{\href}[2]{#2}

\end{document}